\documentclass[12pt,reqno]{article}

\usepackage[usenames]{color}
\usepackage{amssymb}
\usepackage{graphicx}
\usepackage{amscd}

\usepackage{fca}

\usepackage[colorlinks=true,
linkcolor=webgreen,
filecolor=webbrown,
citecolor=webgreen]{hyperref}

\definecolor{webgreen}{rgb}{0,.5,0}
\definecolor{webbrown}{rgb}{.6,0,0}

\definecolor{lightgray}{rgb}{0.827, 0.827, 0.827}

\usepackage{color}
\usepackage{fullpage}
\usepackage{float}

\usepackage{graphics,amsmath,amssymb}
\usepackage{amsthm}
\usepackage{amssymb,amsmath}
\usepackage{amsfonts}
\usepackage{latexsym}
\usepackage{epsf}

\newcommand{\seqnum}[1]{\href{http://oeis.org/#1}{\underline{#1}}}

\begin{document}

\begin{center}
\epsfxsize=4in
\end{center}

\theoremstyle{plain}
\newtheorem{theorem}{Theorem}
\newtheorem{corollary}[theorem]{Corollary}
\newtheorem{lemma}[theorem]{Lemma}
\newtheorem{proposition}[theorem]{Proposition}

\theoremstyle{definition}
\newtheorem{definition}[theorem]{Definition}
\newtheorem{example}[theorem]{Example}
\newtheorem{conjecture}[theorem]{Conjecture}

\theoremstyle{remark}
\newtheorem{remark}[theorem]{Remark}
\newcommand{\red}{{\mbox{red}}}
\newtheorem{algo}{Algorithm}

\begin{center}
\vskip 1cm{\LARGE\bf On the Cryptomorphism between Davis' Subset Lattices, Atomic Lattices, and Closure Systems under T1 Separation Axiom 
} \\
\vskip 1cm 
\large
Dmitry I. Ignatov\\
HSE University, Russia\\
\href{mailto:dignatov@hse.ru}{\tt dignatov@hse.ru}
\end{center}
\vskip .2 in

\begin{abstract} In this paper we count set closure systems (also known as Moore families) for the case when all single element sets are closed. In particular, we give the numbers of such strict (empty set included) and non-strict families for the base set of size $n=6$. We also provide the number of such inequivalent Moore families with respect to all permutations of the base set up to $n=6$. The search in OEIS and existing literature revealed the coincidence of the found numbers with the entry for D.\ M.~Davis' set union lattice (\seqnum{A235604}, up to $n=5$) and $|\mathcal L_n|$, the number of atomic lattices on $n$ atoms, obtained by S.\ Mapes (up to $n=6$), respectively. Thus we study all those cases, establish one-to-one correspondences between them via Galois adjunctions and Formal Concept Analysis, and provide the reader with two of our enumerative algorithms as well as with the results of these algorithms used for additional tests. Other results include the largest size of intersection free families for $n=6$ plus our conjecture for $n=7$, an upper bound for the number of atomic lattices $\mathcal L_n$, and some structural properties of $\mathcal L_n$ based on the theory of extremal lattices.

\end{abstract}

\section{Introduction}\label{sec:intro}

It is known that closure systems (also known as Moore families after E. H. Moore~\cite{Moore:1910}) and complete lattices are closely interconnected \cite{GanterWille:1999,Caspard:2003}. The subject of our study is the connection between a special class of Moore families on  set of $n$ elements, in which every single element set is closed (w.r.t. the so-called T1 separation axiom~\cite[p.126]{Gratzer:2009}), atomic lattices~\cite{Mapes:2010}, and union-closed families studied by D.\ M.~Davis with a certain topological and combinatorial interest~\cite{Davis:2013}. However, the main computational result of the paper is the contribution to OEIS by a new sequence containing the number of inequivalent closure systems under T1 up to $n=6$ and extending its labelled strict and non-strict versions by their common 6th member.

During the experiments on computing the 6th member of \seqnum{A334254} and \seqnum{A334255} with our closure system enumeration algorithm, we found out that their inequivalent counterpart coincides with all the known members (up to $n=5$) of  \seqnum{A235604} except the second member of \seqnum{A334255}. The literature search also revealed coincidence of the 5th and 6th members of \seqnum{A334254}  with that of the number of atomic lattices on five and six elements, respectively; the numbers are reported by S. Mapes~\cite{Mapes:2010}.

Our algorithmic solution is based on one-to-one correspondences between the studied objects, namely, Davis' lattices, atomic lattices, closure system under T1 separation axiom, and certain reduced formal contexts (which can be seen as minimal binary relations) that give rise to appropriate Galois connections and adjunctions.  Galois connections and their resulting Galois lattices (or concept lattices) represented by two dually isomorphic lattices of closed sets  were extensively studied in Formal Concept Analysis~\cite{GanterWille:1999,Caspard:2003,Caspard:2012}, an applied branch of modern lattice theory suitable for data analysis and knowledge processing, while Galois adjunctions are rather related to union-closed systems and were studied with regards to applications in category theory and topology~\cite{Erne:1993,Erne:2004,Keimel:2014}.

Our by-product is related to the problem of the maximal size union-free set family (the asymptotic was given by D. J. Kleitman~\cite{Kleitman:1976}), which is dually equivalent to the problem of maximal size intersection-free family or the maximal size of a reduced formal context on $n$ objects as noted by B.~Ganter and R.~Wille~\cite{GanterWille:1999}. We have found the value of the latter sequence for $n=6$ and have provided a concrete lower bound for this value in the case $n=7$. 

The paper is organized as follows. Section~\ref{sec:defs} contains main definitions from lattice theory and studied papers by Davis~\cite{Davis:2013} and Mapes~\cite{Mapes:2010}. Section~\ref{sec:crypto} establishes correspondence between the three considered problems via Galois adjunctions. Section~\ref{sec:algo} describes two modifications of the \textsc{AddByOne} algorithm to enumerate both labelled and inequivalent families, respectively. In Section~\ref{sec:results}, we summarize the main results obtained. Section~\ref{sec:alter} discusses some alternative approaches that we used (or which can be potentially exploited) for the additional tests with contributed sequences.

\section{Main definitions}\label{sec:defs}

In this section, we mainly use basic definitions and propositions from the book by Ganter and Wille on Formal Concept Analysis~\cite{GanterWille:1999}; these basic notions and facts can also be found in classic and recent monographs on lattice theory as well~\cite{Birkhoff:1967,Caspard:2012,Gratzer:2009}. 

\subsection*{Lattices}

A {\it lattice} is a partially ordered set $\mathbf{L}:=(L, \leq)$ such that for every pair of its elements $x$ and $y$,   the supremum $x \vee y$ and infimum $x \wedge y$ always exist. $(L, \leq)$ is called a {\it complete lattice}, if the supremum $\bigvee X$ and the infumum $\bigwedge X$ exist for any subset $X$ of $L$. Every complete lattice $\mathbf{L}$ has its largest element $\bigvee L$ called the {\it unit element} of the lattice and denoted by $\mathbf 1_L$. Dually, the smallest element of any complete lattice $\mathbf 0_L$ is called the {\it zero element}. 

\begin{lemma}\label{lemma:infcomp}

Any ordered set in which the infimum exists for every subset is a {\it complete lattice}. 

\end{lemma}

The upper neighbours of the zero element (if they exist) are called {\it atoms} of the lattice; dually, the lower neighbors of the unit element are called {\it coatoms}. 

An {\it atomic lattice} (some authors prefer the term atomistic like Ganter and Wille~\cite{GanterWille:1999}) is a complete lattice where each its element is the supremum of atoms.

\subsection*{Closure systems and operators}

In what follows, to consider various set systems, without loss of generality we mainly use the set of first $n$ natural numbers instead of an arbitrary finite set of the same cardinality. We also use $[n]$ as a shorthand for the set of elements $\{1,2,\ldots n\}$.

A {\it closure system} on a set $[n]$ is a set of its subsets which contain $[n]$ and is closed under intersection. That is $\mathcal M \subseteq 2^{[n]}$ is a closure system if $[n] \in \mathcal M$ and $$\mathcal X \subseteq \mathcal M \Rightarrow \bigcap \mathcal X \in \mathcal M.$$

If a closure system $\mathcal M$ contains emptyset, then $\mathcal M$ is {\it strict}.

A {\it closure operator} $\varphi$ on $[n]$ is a map assigning a closure $\varphi X \subseteq [n]$ to each subset $X \subseteq [n]$ under the following conditions:

\begin{enumerate}

\item  $X \subseteq Y \Rightarrow \varphi X \subseteq \varphi Y$ \hfill (monotony)
\item $X \subseteq \varphi X $ \hfill (extensity)
\item $\varphi\varphi X = \varphi X$ \hfill (idempotency) 
\end{enumerate}

T1 separation axiom for a closure system $\mathcal M$ over $[n]$ states that every single element set $\{i\} \in [n]$ is in $\mathcal M$, or, equivalently, is closed, i.e. $\varphi \{i\}= \{i\}$~\cite{Gratzer:2009}.

Every closure system $\mathcal M \subseteq 2^{[n]}$ defines a closure operator as follows:
$$\varphi_{\mathcal M} X:= \bigcap \{A \in \mathcal M \mid X \subseteq	 A\}.$$ While the set of closures of a closure operator $\varphi$ is always a closure system $\mathcal M_\varphi$.

\subsection*{Davis' lattice}

We keep the original notation of Davis~\cite{Davis:2013} in this subsection whenever it is possible.

\begin{definition}

 If $\mathbb M = \{X_1,\ldots,X_n\}$ is a collection of sets, and $S \subseteq [n]$, let
 $$\mathbb M_S :=\bigcup\limits_{i\in S} X_i.$$
  
The set $\mathbb  M$ is called {\it proper} if it is never the case that $X_i \subseteq X_j$ for $i\neq j$. Any $\mathbb  M$ defines a lattice $L(\mathbb M)$ on $2^{[n]}$ by $S\leq T$ if $\mathbb M_S \subseteq \mathbb  M_T$. Lattices $L$ and $L'$ on $2^{[n]}$ are said to be equivalent if there is a permutation $\pi$ of $[n]$ under which the induced permutation of $2^{[n]}$ preserves the lattice relations; i.e., $\pi(S) \leq \pi(T)$ iff $S \leq T$. 

\end{definition}

As  Davis states: ``For a possible application to algebraic topology, we have become interested in an enumeration problem for lattices of subsets, which we have been unable to find in the literature''. Moreover, he invites: ``We wish to introduce it for further investigation.'' 

\section{Establishing cryptomorphisms}\label{sec:crypto}

\subsection{Theory}

As it is shown by F. Domenach in~\cite{Domenach:2013}, different lattice cryptomorphisms can be established by using common lattice properties and various binary relations to enable usage of the defined notions interchangeably. Below, we establish connections between the three studied algebraic structures in a similar fashion.

\begin{theorem}\label{thm:sepiso}

Let $\mathcal M \subseteq 2^{[n]}$ be a strict closure system with T1 separation axiom fulfilled, then $(\mathcal M, \subseteq)$ is an atomic lattice with $\bigwedge \mathcal X = \bigcap \mathcal X$ and $\bigvee \mathcal X = \varphi_{\mathcal M} \bigcup \mathcal X$ for all $\mathcal X \subseteq \mathcal M$. Conversely, every atomic lattice is isomorphic to the lattice of all closures of a strict closure system  with T1 separation axiom fulfilled.

\end{theorem}

\begin{proof} The infimum of $\mathcal X$ in $(\mathcal M, \subseteq)$ is defined as $\bigwedge \mathcal X = \bigcap \mathcal X$. Since all single element sets  are closed, i.e., $\{i\} \in \mathcal M$, then 
they are the only upper neighbors of $\emptyset$, which is the zero element of $(\mathcal M, \subseteq)$.

 By Lemma~\ref{lemma:infcomp}, there exists the supremum of $\mathcal X$, which is defined as  $\bigvee \mathcal X = \varphi_{\mathcal M} \bigcup \mathcal X$. Since $\bigvee \mathcal X =  S$ for some  $S \in \mathcal M$, then $S=\varphi_{\mathcal M} \{S\}=\varphi_{\mathcal M} \bigcup\limits_{s \in S} \{s\}= \bigvee\limits_{s \in S} \{s\}$.

Let $\mathbf L=(L,\leq)$ be an atomic lattice on $n$ atoms. Then the set system $\{(x]\setminus \mathbf 0_{L} \mid x \in L\}$ is a strict closure system under T1 axiom since $\bigcap\limits_{y \in T}(y]\setminus   \mathbf 0_{L}=(\bigwedge T ] \setminus  \mathbf 0_{L}$, the system contains empty set and $n$ single-element sets obtained from each atom of $\mathbf L$, respectively.

\end{proof}

\begin{proposition}\label{prop:strictcl}

Every closure system  $\mathcal M \subseteq 2^{[n]}$ with T1 separation axiom fulfilled is strict for $n \neq 1$.

\end{proposition}

\begin{proof} For $n=0$ the proposition holds trivially. For $n=1$ the system $\{\{1\}\}$ is not strict. For $n\geq 2$ any pair $i,j \in [n]$ implies $\{i\} \cap \{j\} =\emptyset$; hence $\emptyset \in \mathcal M$.

\end{proof}

To deal with Davis' lattice, which in fact combines two isomorphic lattices, let us reformulate the original definition.

Let $U=\bigcup\limits_{i \in [n]} X_i$ and $R$ be a binary relation on $[n] \times U$ with $i R u$ if  $u \in X_i$.

Consider two operators, $(\cdot)^\cup: 2^{[n]} \to 2^U$ and $(\cdot)^\subseteq: 2^U \to 2^{[n]}$ that are defined as follows for any $A \subseteq [n]$ and $B \subseteq U$:

$$A^\cup:=\{u \mid iRu \mbox{ for some } i \in A \}$$

\noindent (the union of all $X_i$ with $i \in A$, i.e., $\mathbb M_A$ in Davis' notation)
$$B^\subseteq:=\{i \mid iRu \mbox{ implies } u \in B \}$$

\noindent (all indices $i$ such that $X_i \subseteq B$).

These two operators $((\cdot)^\cup,(\cdot)^\subseteq)$ forms the so-called {\it axialities} (cf.\ Birkhoff's {\it polarities}~\cite{Birkhoff:1967}), i.e., Galois adjunction~\cite{Erne:1993,Erne:2004} between powersets of $[n]$ and $U$. Note that Galois adjunctions between ordered sets are also known as isotone Galois connections~\cite{Keimel:2014}.

Before we proceed with formal definitions and proofs, let us consider properties of the composite operators $(\cdot)^{\cup\subseteq}: 2^{[n]} \to 2^{[n]}$  and  $(\cdot)^{\subseteq\cup}: 2^U \to 2^U$.

\begin{proposition}\label{prop:prop} Let $R\subseteq	[n] \times U$ is binary relation, $A, A_1, A_2 \subseteq [n]$ and $B, B_1, B_2 \subseteq U$, then

\begin{enumerate}

\item a) $A_1 \subseteq A_2 \Rightarrow A_1^\cup \subseteq A_2^\cup$ and b) $B_1 \subseteq B_2 \Rightarrow B_1^\subseteq \subseteq B_2^\subseteq$

\item a) $(\cdot)^{\cup\subseteq}$ is conrtactive, i.e., $(A)^{\cup\subseteq} \subseteq A$, while b) $(\cdot)^{\subseteq\cup}$ is extensive

\item  a) $A^{\cup}=A^{\cup\subseteq\cup}$ and b) $B^{\subseteq}=B^{\subseteq\cup\subseteq}$

\item $(\cdot)^{\cup\subseteq}$ and $(\cdot)^{\subseteq\cup}$ are idempotent

\item $(\cdot)^{\cup\subseteq}$ and $(\cdot)^{\subseteq\cup}$ are isotone.

\end{enumerate}

\end{proposition}

\begin{proof}

 1. a) $A_2^\cup=\{u \mid iRu \mbox{ for some } i \in A_1 \cup (A_2\setminus A_1) \}=A_1^\cup \cup \{u \mid iRu \mbox{ for some } i \in A_2\setminus A_1\}$. 1. b) If $i \in B_1^\subseteq$ then $iRu$ implies $u \in B_1$, i.e.,  also $u \in B_2$ since $B_1 \subseteq B_2$.

\noindent 2. a) If $j \in A^{\cup\subseteq}$, then $jRu$ implies $u \in A^\cup$, i.e., there exists $i \in A$ such that $iRu$. In short, $jRu$ implies $iRu$. 2. b) Let $u \in B$, then all $i$ such that $iRu$ are in  $B^\subseteq$. It means that $u$ is also in $B^{\subseteq\cup}$ since $(\cdot)^\cup$ collects all $v$ incident to $i$ in $R$.

\noindent 3. a) $ A^{\cup\subseteq\cup} \subseteq A^{\cup} $ by 2a and 1a, while $ A^{\cup} \subseteq A^{\cup\subseteq\cup}$ follows immediately from 2b.

\noindent 4.  follows from 3.a and 3.b, respectively.

\noindent 5.  follows from 1.

\end{proof}

\begin{definition}[\cite{Keimel:2014}]

A pair $(\alpha, \beta)$ of maps $\alpha: P\to Q$, $\beta: Q \to P$ is called a  Galois adjunction between the posets $(P,\leq)$ and $(Q,\leq)$ provided that
 $$\mbox{ for all } p\in P \mbox{ and } q\in Q, \mbox{ we have } \beta q \leq p \iff q\leq\alpha p .$$

\end{definition}

Given Galois adjunction $(\alpha, \beta)$, $\alpha$ is called the upper adjoint of  $\beta$ and  $\beta$ the lower adjoint of $\alpha$.

\begin{theorem} The pair of operators $((\cdot)^\cup,(\cdot)^\subseteq)$ forms  the Galois adjunction between powersets of $[n]$ and $U$ related by $R \subseteq [n] \times U$.

\end{theorem}

\begin{proof} For $A \subseteq [n]$ and $B \subseteq U$ we need to prove $B^\subseteq \subseteq A \iff B \subseteq A^\cup$.

\noindent $\Rightarrow$ Let $B^\subseteq \subseteq A$, then due to isotony  of $(\cdot)^\cup$ (Proposition~\ref{prop:prop}.1a) we have $B^{\subseteq\cup} \subseteq A^\cup$ and by extensity of $(\cdot)^{\subseteq\cup}$  we get $B \subseteq B^{\subseteq\cup}$.

\noindent $\Leftarrow$ Similarly, due to isotony  of $(\cdot)^\subseteq$ (Proposition~\ref{prop:prop}.1b) we have $B^{\subseteq} \subseteq A^{\cup\subseteq}$ and by contraction of $(\cdot)^{\subseteq\cup}$  we get $A^{\cup\subseteq} \subseteq A$.

\end{proof}

Note that any contractive, monotone, and idempotent operator $\psi$ on a set $S$ is called {\it kernel (interior) operator}. Its {\it fixed points}, i.e.,  $X \subseteq S$ such that $\psi X=X$ are called {\it open sets} or {\it dual closures}.

\begin{theorem}[\cite{Caspard:2012},\cite{Erne:2004}] \label{thm:iso}
Let $(P,\leq)$ and $(Q,\leq)$ be two ordered sets, $(\alpha, \beta)$ is the Galois adjunction between them. The following properties hold:

\noindent 1. $\alpha\beta\alpha =\alpha$ and $\beta\alpha\beta =\beta$.

\noindent 2. The composition map $\varphi = \beta\alpha$ is a kernel operator on $P$ and the composition map  $\psi = \alpha\beta$  is a closure operator on $Q$.

\noindent 3. The ordered subset $\varphi(P)$ of open elements of $\varphi$ in $P$ is equal to $\beta(Q)$ and the ordered subset $\psi(Q)$ of closed elements of $\psi$ in $Q$ is equal to $\alpha(P)$. The ordered subsets $\varphi(P)$ and $\psi(Q)$ are isomorphic, by the restrictions to the latter of $\alpha$ and $\beta$.

\end{theorem}

Particular cases of Statements 1. and 2. of Theorem~\ref{thm:iso} are proven in Proposition~\ref{prop:prop} for the pair of operators $((\cdot)^\cup,(\cdot)^\subseteq)$. Statement 3 implies the following corollary.

\begin{corollary}\label{cor:iso} For a given binary relation $R \subseteq [n] \times U$ and the Galois adjunction $((\cdot)^\cup,(\cdot)^\subseteq)$,  the ordered set of  dual closures  $\mathcal K_{\cup\subseteq}=\{ A^{\cup\subseteq} \mid A \subseteq [n]\}$ is equal to $\mathcal K_{\subseteq}=\{ B^\subseteq\mid B \subseteq U\}$ and the ordered  set of  closures $\mathcal M_{\subseteq\cup}=\{ B^{\subseteq\cup} \mid B \subseteq U\}$ is equal to $\mathcal M_{\cup}=\{ A^\cup\mid A \subseteq [n]\}$. $K_{\cup\subseteq}$ and $\mathcal M_{\subseteq\cup}$ are isomorphic subject to restrictions of $(\cdot)^\cup$ and $(\cdot)^\subseteq$.

\end{corollary}

Let us consider $L(\mathbb M)$ on a proper set $\mathbb M =\{X_1, X_2, \ldots, X_k\}$ (i,e., the antichain of $X_i, i\in\{1,\ldots, k\}$ ) and the related incidence relation $R \subseteq [n] \times U$ with the Galois adjunction $((\cdot)^\cup,(\cdot)^\subseteq)$.

\begin{theorem}\label{thm:atoMoore}
$L(\mathbb M)=(\mathcal M_{\subseteq\cup}, \subseteq)$ is atomic lattice. Conversely, every atomic lattice is isomorphic to some $(\mathcal M_{\subseteq\cup}, \subseteq)$.
\end{theorem}

\begin{proof}  $\Rightarrow$ $L(\mathbb M)=(\mathcal M_{\subseteq\cup}, \subseteq)$ since, by Corollary~\ref{cor:iso}, $(\mathcal M_{\subseteq\cup}, \subseteq)=(\mathcal M_{\cup}, \subseteq)$ and  $A^\cup= \mathbb M(A)$ for $A \subseteq M$, by the definitions of operators $(\cdot)^\cup$ and $\mathbb M(\cdot)$. The zero element of $L(\mathbb M)$ is $\mathbf 0_{L(\mathbb M)}=\mathbf 0_{\mathcal M_{\subseteq\cup}}=\emptyset$. The upper neighbors of $\mathbf 0_{L(\mathbb M)}$ are closed single-element sets $\{i\}=\{i\}^{\cup\subseteq}=\mathbb M(\{i\})$ for $i \in [n]$ and every $A \in \mathcal M_{\subseteq\cup}$ is equal to $\bigcup\limits_{i \in A}\{i\}$. 

\noindent 
$\Leftarrow$  By Theorem~\ref{thm:sepiso}. Since every single-element set is closed in  $\mathcal M_{\subseteq\cup}$ (T1 separation axiom is fulfilled) and $\emptyset^{\cup\subseteq}=\emptyset$ (the system $\mathcal M_{\subseteq\cup}$ is strict). Note that for $n=1$, we have the only system $\mathcal M_{\subseteq\cup}=\{\emptyset,\{1\}\}$.
\end{proof}

Theorem~\ref{thm:atoMoore} allows us to transfer atomicity to Moore families. Then let us call any Moore family containing all single element sets from its base set \textit{atomic} or, equivalently, \textit{atomic closure system}.

Corollary~\ref{cor:iso} implies that kernel system $\mathcal K_{\cup\subseteq}$ is also isomorphic to $L(\mathbb M)$. Actually, one can pair the fixed points of $(\cdot)^{\cup\subseteq}$ and $(\cdot)^{\subseteq\cup}$ via the Galois adjunction as follows.

For a given Galois adjunction $(\alpha, \beta)$ between two ordered sets $(P,\leq)$ and $(Q,\leq)$, consider a pair $( p,q)$, where $p \in P$ and $q \in Q$ and  $p=\beta  q$ and $q= \alpha p$.

In case of the adjunction on a binary set $R\subseteq [n] \times U$ for Davis' lattice, such pairs are called \textit{upper concepts} \cite{Wolski:2004,Caspard:2012}, and we have the \textit{lattice of upper concepts} $(\mathcal A, \sqsubseteq)$ such that 

$$\mathcal A =\{ (A,B) \mid  A^\cup= B \mbox{ and } B^\subseteq=A \mbox{ for } A\subseteq [n], B \subseteq U\}$$

and

$$(A,B) \sqsubseteq (C,D) \iff A \subseteq C \mbox{ and } B \subseteq D \mbox{ for } (A,B), (C,D) \in \mathcal A.$$

We provide several examples of such lattices in the next subsection.

One more ismorphism exists between atomic and LCM lattices, where LCM stands for least common multiplier; see, for example, S. Mapes work~\cite{Mapes:2010}.

\subsection{Examples}

Let us consider several binary relations and the lattices of their lower concepts. In Fig.~\ref{fig:context}, one can see three $3 \times 3$ exemplary binary relations often used in Formal Concept Analysis for data scaling~\cite{GanterWille:1999}.

\begin{figure}[ht]\label{fig:context}
  \centering
\begin{minipage}{0.3\linewidth}

\begin{center}
\begin{cxt}%
\cxtName{}%
\att{$a$}%
\att{$b$}%
\att{$c$}%
\obj{xxx}{1}
\obj{xx.}{2}
\obj{x..}{3}
\end{cxt}
\end{center}

\end{minipage}
\begin{minipage}{0.3\linewidth}

\begin{center}
\begin{cxt}%
\cxtName{}%
\att{$a$}%
\att{$b$}%
\att{$c$}%
\obj{x..}{1}
\obj{.x.}{2}
\obj{..x}{3}
\end{cxt}
\end{center}

\end{minipage}
\begin{minipage}{0.3\linewidth}

\begin{center}
\begin{cxt}%
\cxtName{}%
\att{$a$}%
\att{$b$}%
\att{$c$}%
\obj{.xx}{1}
\obj{x.x}{2}
\obj{xx.}{3}
\end{cxt}
\end{center}

\end{minipage} 
 \caption{Example relations for order, nominal, and contranominal scales.}
  \label{fig:scales}
\end{figure}

The line (or Hasse) diagrams of the relations are shown in Fig.~\ref{fig:diag}. The shaded nodes depict atoms of the lattices. The leftmost lattice is not atomic since it is a chain of four elements.

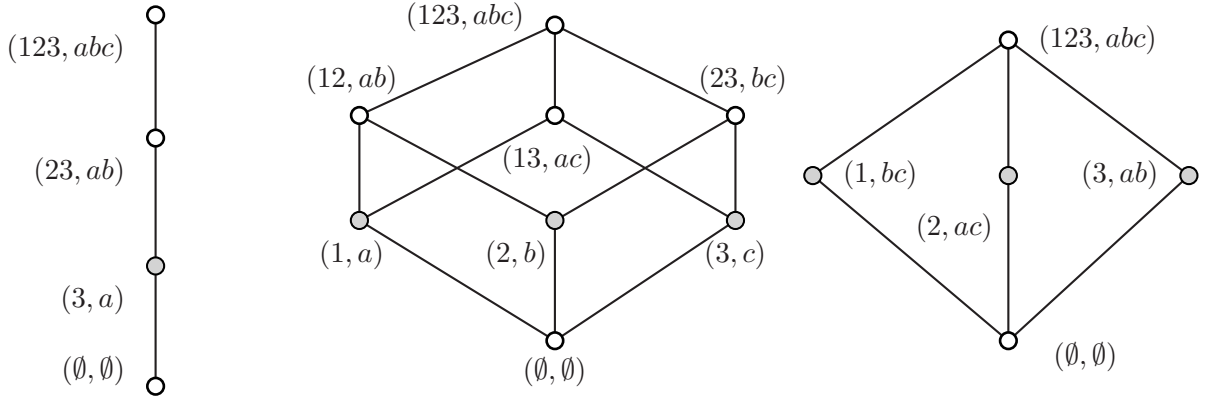
\begin{figure}[ht]\label{fig:diag}
\begin{minipage}[h]{0.3\linewidth}  
  \centering
	\begin{picture}(100,140)
\unitlength 0.20mm
\begin{diagram}{180}{260}
\Node{1}{100}{252}
\Node{2}{100}{170}
\Node{3}{100}{85}
\Node{4}{100}{5}
\Edge{1}{2}
\Edge{2}{3}
\Edge{3}{4}
\NoDots
\leftObjbox{1}{20}{13}{$(123,abc)$}
\leftObjbox{2}{20}{13}{$(23,ab)$}
\leftObjbox{3}{20}{13}{$(3,a)$}
\leftObjbox{4}{20}{-20}{$(\emptyset,\emptyset)$}
\CircleSize{11}
\end{diagram}
\put(-80,85){\color{lightgray}{\circle*{10}}}%
\end{picture}
 \end{minipage} 
\hspace{2mm}
\begin{minipage}[h]{0.3\linewidth}
\begin{picture}(100,100)
\unitlength 0.20mm
\begin{diagram}{250}{210}
\Node{1}{130.0}{210}
\Node{2}{0}{150.0}
\Node{3}{130}{150.0}
\Node{4}{250}{150.0}
\Node{5}{0.0}{80.0}
\Node{6}{130.0}{80.0}
\Node{7}{250.0}{80.0}
\Node{8}{130.0}{0.0}
\Edge{1}{2}
\Edge{1}{3}
\Edge{1}{4}
\Edge{2}{5}
\Edge{2}{6}
\Edge{3}{5}
\Edge{3}{7}
\Edge{4}{6}
\Edge{4}{7}
\Edge{5}{8}
\Edge{6}{8}
\Edge{7}{8}
\NoDots
\leftObjbox{1}{20}{-15}{$(123,abc)$}
\centerObjbox{2}{-5}{-33}{$(12,ab)$}
\centerObjbox{3}{-5}{20}{$(13,ac)$}
\centerObjbox{4}{5}{-33}{$(23,bc)$}
\centerObjbox{5}{-5}{13}{$(1,a)$}
\leftObjbox{6}{5}{13}{$(2,b)$}
\centerObjbox{7}{0}{13}{$(3,c)$}
\centerAttbox{8}{0}{-30}{$(\emptyset,\emptyset)$}
\CircleSize{11}
\end{diagram}
\put(0,80){\color{lightgray}{\circle*{10}}}%
\put(-120,80){\color{lightgray}{\circle*{10}}}%
\put(-250,80){\color{lightgray}{\circle*{10}}}%
\end{picture}
  \end{minipage} 
  \hspace{8mm}
\begin{minipage}[h]{0.3\linewidth}
\begin{picture}(100,100)
  \unitlength 0.20mm
\begin{diagram}{150}{200}%
\Node{1}{130.0}{200}%
\Node{2}{0}{110.0}%
\Node{3}{130}{110.0}%
\Node{4}{250}{110.0}%
\Node{5}{130.0}{0.0}%
\Edge{1}{2}%
\Edge{1}{3}%
\Edge{1}{4}%
\Edge{2}{5}%
\Edge{3}{5}%
\Edge{4}{5}%
\NoDots%
\rightObjbox{1}{20}{-10}{$(123,abc)$}%
\rightObjbox{2}{20}{-10}{$(1,bc)$}%
\leftObjbox{3}{10}{25}{$(2,ac)$}%
\leftObjbox{4}{20}{-10}{$(3,ab)$}%
\rightAttbox{5}{30}{-20}{$(\emptyset,\emptyset)$}%
\CircleSize{11}
\end{diagram}
\put(100,110){\color{lightgray}{\circle*{10}}}%
\put(-20,110){\color{lightgray}{\circle*{10}}}%
\put(-150,110){\color{lightgray}{\circle*{10}}}%
\end{picture}
  \end{minipage} 
 \caption{The line diagrams of the lattices of lower concepts for the binary relations in Fig.~\ref{fig:context}, from  left to right, respectively.}
  \label{fig:posneglatts}
\end{figure}

From Mapes~\cite{Mapes:2010}, we know that $\mathcal L_n$, the lattice of all atomic lattices is also an atomic lattice. The central diagram in Fig.~\ref{fig:diag} shows the  unit element (the Boolean lattice on three elements) of $\mathcal L_3$, while the rightmost diagram shows its zero element (the diamond $M_3$), w.r.t. to the established isomorphism of lattices. 

\section{Algorithms}\label{sec:algo}

This section introduces the algorithm to traverse and count all binary relations resulting in unique atomic lattices.

First, we note that system of sets is proper (in the sense of Davis) if and only if each element in closed. 

\begin{theorem}\label{thm:ac2sep} Let $R \subseteq [n] \times U$ be a binary relation with the Galois adjunction $((\cdot)^\cup,(\cdot)^\subseteq)$. 1) For every pair $i,j \in [n]$ such that $i\neq j$, it fulfills $\{i\}^\cup\not\subseteq \{j\}^\cup$ (antichain condition) $\iff$  2)  for $i \in [n]$  $\{i\}$ is closed w.r.t. set intersection (T1 separation axiom) $\iff$  3) for $i \in [n]$ $\{i\}$ is dually closed, i.e. $\{i\}^{\cup\subseteq}=\{i\}$.

\end{theorem}

\begin{proof} $1 \Rightarrow 3$. If $\{i\}^\cup\not\subseteq \{j\}^\cup (i \neq j)$, then $\{i\}^{\cup\subseteq}=\{k |\{k\}^\cup \subseteq \{i\}^\cup\}=\{i\}$.

\noindent $3 \Rightarrow 2$. If  $\{i\}^{\cup\subseteq}=\{i\}$, then there is no $j\neq i$ such that $\{j\}^{\cup}  \subseteq \{i\}^{\cup} $. Hence, there is a unique $S \subseteq [n]$ such that $\{i\}^\cup=\bigcap\limits_{k \in S\subseteq [n]} \{k\}^\cup$, namely, $S=\{i\}$.

 \noindent  $2 \Rightarrow 1$. If $\{i\}$ cannot be intersection of any $S \subseteq [n]$ except $S=\{i\}$, then $\{i\}^\cup \not \subseteq \bigcap\limits_{j\in S} \{j\}^\cup$ for every $S \subseteq [n]$ such that $i \not\in S$. Hence, there is no $j\neq i$ such that $\{j\}^{\cup}  \subseteq \{i\}^{\cup} $ and $\{i\}^\cup$ cannot be union of sets $T \subseteq  U$ except $T=\{i\}^\cup$, which implies $\{i\}^{\cup\subseteq}=\{i\}$.

\end{proof}

Note the antichain condition also implies absence of duplicate rows in $R$ being represented as an incidence table since $\{i\}^\cup$ is ``the row'' of related elements to $i$ by $R$.

\begin{theorem}\label{thm:clo_dclo} Let $R \subseteq [n] \times U$ be a binary relation with the Galois adjunction $((\cdot)^\cup,(\cdot)^\subseteq)$.  $A\subseteq [n]$ is open (dually closed), $\{A\}^{\cup\subseteq}=\{A\}$ $\iff$ $ A$ is closed w.r.t. set intersection in the complementary binary relation $\overline R$ such that $(i,u) \in \overline{R} \iff (i,u) \not \in R$.

\end{theorem}
\begin{proof}

Let $\{i\}'=\overline{\{i\}^\cup}$, i.e. $m \in \{i\}'\iff  m \not\in \{i\}^\cup$.

$\Rightarrow$ Taking into account contraposition, for any dually closed $A \subseteq [n]$ we get

$$\overline {A^{\cup}}=\overline{\bigcup\limits_{i\in A} \{i\}^\cup}=\bigcap\limits_{i \in A} \overline{\{i\}^\cup}=\bigcap\limits_{i \in A} \{i\}' \mbox{ , i.e.,}$$

\noindent the intersection of rows  in $\overline R$ with all row indices from $A$. There are no other $j \in [n]\setminus A$ with $\bigcap\limits_{i \in A} \{i\}' \subseteq \{ j\}'$ since $A$ is open, i.e. $\bigcup\limits_{i \in A} \{i\}^\cup \not\supseteq \{ j\}^\cup$.

\noindent  $\Leftarrow$ Since $A$ is closed, then there exists $B=\bigcap\limits_{i \in A} \{i\}'$, i.e.\ intersection of all rows in $\overline{R}$ with indices from $A$. By contraposition, we get 

$$\overline {B}=\overline{\bigcap\limits_{i\in A} \{i\}^\prime}=\bigcup\limits_{i \in A} \overline{\{i\}^\prime}=\bigcup\limits_{i \in A} \{i\}^\cup \mbox{ , i.e.,}$$

\noindent the union of rows  in $R$ with all row indices from $A$.

There are no other $j \in [n]\setminus A$ with $ \{ j\}^\cup \subseteq \overline{B} $ since $A$ is closed, i.e. $\bigcap\limits_{i \in A} \{i\}^\prime \not\subseteq \{ j\}^\prime$.

\end{proof}

Theorems~\ref{thm:ac2sep} and \ref{thm:clo_dclo} allow us working with enumeration of all closure systems and that of all kernel systems or all systems of sets closed under union interchangeably, given fixed $n$ and the smallest $R$ w.r.t the size of $U$. Similar replacement of Moore families by set systems closed under union was exploited by Brinkmann and Deklerck~\cite{Brinkmann:2018}.

Since different binary relations $R \subseteq [n] \times U$ with a fixed $n$ can produce the same Moore families on $[n]$ (e.g., by removing a full column $[n] \times \{u\}$ in $R$ if  $[n] \times \{u\} \subseteq R$), we need to identify valid ways to reduce $U$ and thus $R$ without affecting the resulting Moore family.

\begin{definition}[adopted from Ganter and Wille \cite{GanterWille:1999}] We call a binary relation $R \subseteq [n] \times U$ \textit{column reduced} if 1) it is clarified, i.e.\ $R$ does not contain duplicate rows and columns ($\forall i,j \in [n]:  \{i\}^\cup=\{j\}^\cup \Rightarrow i=j$; similarly, for $u,v \in U$) and 2) there is no $u \in U$, which can be obtained by intersection of other columns $X \subseteq U$, i.e.\ $u \not \in X$ and  $\bigcap\limits_{x \in X} x^\subseteq\neq u^\subseteq$. 
\end{definition}

A  \textit{row reduced} binary relation is defined similarly. If $R$ is both row and column reduced, $R$ is called \textit{reduced}.

In practice, we cannot simultaneously eliminate all the rows and the columns that are \textit{reducible}, but this is no problem if we add rows (or columns) in a lectic order and check reducibility. 

By Sperner theorem~\cite{Sperner:1928} the largest set antichain in $2^{[n]}$ contains $\binom{n}{\lfloor n/2 \rfloor}$ sets. It makes it possible to deduce the exact lower bound for the number of elements in $U$ for Davis' lattice and the associated relation.

\begin{theorem}\label{thm:lowerbound} The smallest size of $U$ in $R \subseteq	[n] \times U$ such that the associated closure system is $T1$-separated (or atomic) is the minimal $k$ under which $\binom{k}{\lfloor k/2 \rfloor} \geq n$.

\end{theorem}

In our algorithms, we use binary set representation and work with integer types. For example, $7_{10}=111_2$ represents the set $\{1,2,3\}$ (or $\{a,b,c\}$) since the first three bits are on.

The \textsc{Atomic AddByOne} algorithm is inspired by the CloseByOne algorithm proposed by S.\ O.\ Kuznetsov~\cite{Kuznetsov:1993}.

\begin{algo} \textsc{Atomic AddByOne}\label{alg:AtomicAddByOne}

\noindent\textbf{Input}: the number of atoms $n \in \mathbb N$ ($n>1$)

\noindent\textbf{Output}:  the number of Moore families fulfilling T1 separation axiom

\

1.\ Generate all combinations $\binom{2^{[n]}\setminus \{\emptyset,[n]\}}{k_{min}}$ in lectic order.

2.\ Check each combination represented by a tuple $t=(i_1,\ldots,i_k)$ whether it is a column reduced binary relation and fulfils T1 axiom. If yes, store 1 in $cnt[t]$.

3.\ Extend each valid tuple $t$ (the column  reduced binary relation) from step 2 by a next integer $i_{k+1}$ after $i_k$ from $\{i_k+1, \ldots, 2^n-2\}$ and check whether the new tuple $t^*=(i_1,\ldots,i_{k+1})$ is  a reduced binary relation and fulfils T1 axiom (if T1 was fulfilled for $t$, then skip T1-check). If yes, increment $cnt[t]$ and repeat step 3 with $t^*$ recursively.

4.\  Return the sum of all $cnt$-s.

\end{algo}

Step 1 excludes combinations with emptyset since $\emptyset$ should be present in the resulting system as intersection of atoms by Theorem~\ref{thm:sepiso}. Since full rows and full columns are reducible, $2^{n-1}$ is always excluded (every closure system on $[n]$ contains $[n]$ by definition). Note that  all subsets of $2^{[n]}$ of size $k_{min}$, which elements has $\lfloor k_{min}/2 \rfloor$ (or $\lceil k_{min}/2 \rceil$) bits each, forms the antichain of $k_{min}$ elements by Theorem~\ref{thm:lowerbound} and our previous work on Boolean matrix factorization of contranominal scales~\cite{Ignatov:2021}. So, Step 1 can be further improved accordingly for $n$ larger than 6 ($k_{min}=4$).

To enumerate inequivalent atomic Moore families, we apply all the permutations $\pi \in \Pi(n)$ on the set $[n]$ to every subset of a concrete Moore family represented by tuple $t$, i.e.\ we compute all $\pi(t)=(\pi(i_1), \ldots, \pi(i_k))$. We call $t$ \textit{canonic} if it is lectically smallest among all permuted tuples $\pi(t)$. Algorithm~\ref{alg:IneqAtomicAddByOne} counts each canonic representative per an equivalence class w.r.t.\ $\Pi(n)$.

\begin{algo} \textsc{Atomic IneqAddByOne}\label{alg:IneqAtomicAddByOne}

\noindent\textbf{Input}: the number of atoms $n \in \mathbb N$ ($n>1$)

\noindent\textbf{Output}:  the number of inequivalent Moore families fulfilling T1 separation axiom

\

The only modification of Algorithm~\ref{alg:AtomicAddByOne}  is done at step 3.

3'.\ We additionally check whether the new tuple $t^*$ is canonic and count only such tuples.

\end{algo}

All the implementations are coded in Python, speeded up with Cython extension and multiprocess library, and available on the author's Github\footnote{\url{https://github.com/dimachine/ClosureSeparation}} along with the results of experiments recorded in Jupyter notebooks.
 
\section{Resulting numbers and sequences}\label{sec:results}

By means of \textsc{Atomic AddByOne}, we obtained the new 6th member of \seqnum{A334254} and \seqnum{A334255}, i.e.\ enumerated 66960965307  atomic Moore families for $n=6$. The number of strict inequivalent atomic Moore families coincides with all the known members (up to $n=5$) of \seqnum{A235604} and, as implied by theorems~\ref{thm:sepiso} and \ref{thm:atoMoore}, its sixth member 95239971 obtained by us as well.

Since for the 2nd members \seqnum{A334255} and \seqnum{A334254} are different, as well as 2nd member of \seqnum{A235604}  differs from $a(2)$ of the sequence with the numbers of all inequivalent atomic and strict Moore families, we suggest adding a new sequence \seqnum{A355517}, which contains the numbers of strict non-isomorphic atomic Moore families for $n$ up to 6.

\begin{table}
\centering 
\begin{tabular}{|c|c|c|c|c|}
\hline
$n$ & \seqnum{A334254} & \seqnum{A334255} & \seqnum{A235604} & \seqnum{A355517}\\
\hline

0 & 1 & 1 & 1 & 1\\
1 & 2 & 1 & 1 & 2\\
2 & 1 & 1 & 1 & 1\\
3 & 8 & 8 & 4 & 4\\
4 & 545 & 545 & 50 & 50\\
5 & 702 525 & 702 525 & 7 443 & 7 443\\
6 & \it 66 096 965 307 & \it 66 096 965 307  & \it 95 239 971 & \it 95 239 971\\

\hline

\end{tabular}

\caption{Studied sequences with the found extensions in italic.}
\end{table}

In addition to that, we provide a new realization of \seqnum{A305233}. 
Since by Theorem~\ref{thm:lowerbound},  \seqnum{A305233} contains exactly the minimal number $k$ of elements needed to exceed a given number of atoms $n$ with the value of middle binomial coefficient $\binom{k}{\lfloor k/2 \rfloor}$ to form a proper antichain of $n$ sets with $\lfloor k/2 \rfloor$ elements each out of $k$.

Another our  contribution is related to the maximal size of the reduced contexts. 
The first five members for $n=1, \ldots, 5$ are known in the literature~\cite{GanterWille:1999}: 1, 2, 4, 7, 13. While we found the sixth term 24 by full enumeration with Algorithms 1 and 2.

For $n=6$, one out of ten Moore families of length 24 in its equivalence class is as follows:
$$\{7, 11, 13, 14, 19, 21, 22, 25, 26, 29, 30, 37, 38, 39, 41, 42, 43, 44, 49, 50, 51, 52, 56, 60\}.$$

For the 7th member, we state that it is not less than 41, since by combinatorial (though non-exhaustive) search we found the largest  set system of size 41 to form the reduced context of the Moore family: $\{7, 11, 13, 14, 19, 21, 22, 25, 26, 28, 35, 37, 38, 41, 42, 44, 49, 50, 52, 56, 67,$ $69, 70, 73, 74, 76, 81, 82, 84, 88, 97, 98, 100, 104, 113, 114, 116, 121, 122, 123, 124 \}$. In total, 420 Moore families of size 41 are in the found equivalence class. 

Note that the last results for $n=6$ and 7 are also valid for Moore systems without additional constraints.

\section{Other approaches}\label{sec:alter}

There are algorithms to systematically enumerate closed set of any closure operator, in particularly given in terms of Galois connections on a binary relation. To do so in our case for the family of atomic closure systems, which is isomorphic $\mathcal L_6$, we need to know its representation as a binary relation.

In FCA, Ganter and Wille~\cite{GanterWille:1999} showed that for any finite lattice $\mathbf L : = (L, \leq)$, there exists a unique binary relation on join and meet irreducible elements,  $J(\mathbf L)$ and $M(\mathbf L)$, respectively, such that the lattice formed by all its rows (or columns)  closures under intersection is isomorphic to the original lattice; this relation is called a \textit{standard context} and defined as the restriction of $\leq$, i.e.\ as $\leq \cap   J(\mathbf L) \times M(\mathbf L)$. We use $\mathbb K(\mathbf L):=(J(\mathbf L), M(\mathbf L), L), \leq)$ to denote the standard context of a lattice $\mathbf L$. A similar approach to represent and analyze finite lattices based on a poset of irreducibles is employed by G.~Markowsky (e.g., to answer the question from genetics: ``What is the smallest number of factors that can be used to represent a given phenotype system?")~\cite{Markowsky:2019}.

From Mapes we know about two important theorems on the number of  atoms and meet irreducible elements of $\mathcal L_6$. The first theorem on the number of atoms is attributed  to Phan~\cite{Phan:2006} by Mapes~\cite{Mapes:2010}. So, we add the description of atoms to the statement of this theorem, refine the range of valid $n$ and provide its shorter proof here. 

\begin{theorem} The number of atoms of the lattice $\mathcal L_n$ formed by all atomic closure families on $n>1$ is equal to $2^n-n-2$ and each atom has the form $\{\emptyset,\{1\},\ldots,\{n\},\sigma,[n]\}$, where $\sigma \subseteq [n]$ and $2\leq |\sigma| < n $.
\end{theorem} 
\begin{proof} 
For $n=2$ we have no atoms since $\mathbf 0_{\mathcal L_2} = \mathbf 1_{\mathcal L_2}=\{\emptyset,\{1\},\{2\},\{1,2\}\}$. 
The difference of $\mathcal A_\sigma = \{\emptyset,\{1\},\ldots,\{n\},\sigma,[n]\}$  and $\mathbf 0_{\mathcal L_n}=\{\emptyset,\{1\},\ldots,\{n\},[n]\}$ is $\{\sigma\}$ for some $\sigma \subset [n]$ (such that $2\leq |\sigma| < n $). Hence, all $\mathcal A_\sigma$ are the upper neighbors $\mathbf 0_{\mathcal L_n}$ since all other families in $\mathcal L_n$ greater than $\mathbf 0_{\mathcal L_n}$ are  also greater than some $\mathcal A_\sigma$. The number of 
all $\sigma$ is $2^n-n-2$. 

\end{proof}

Note that for $n=1$,  $\mathbf 1_{\mathcal L_1}=\{\emptyset,\{1\}\}$ is the only atom, i.e.\ the upper neighbor of $\mathbf 0_{\mathcal L_1}=\{\{1\}\}$.

The statement of the second theorem is also enriched by us with the description of join-irreducible elements taken directly from the original proof by Mapes~\cite{Mapes:2010} except $n=1$. 

\begin{theorem} Each meet irreducible element in $\mathcal L_n$ for $n>2$ has the form $2^{[n]}\setminus \big[\sigma, [n]\setminus {i}\big]$, where $\sigma \subset [n]$ and $2\leq |\sigma| < n $ and $i \in [n]$. The number meet irreducible elements for $n\neq 1$ is $n(2^{n-1}-n)$ and 1 for $n=1$.
\end{theorem}
\begin{proof} See Mapes~\cite{Mapes:2010}. For $n=1$ the coatom is $\mathbf 0_{\mathcal L_1}=\{\{1\}\}$.

\end{proof}

\begin{corollary} The standard context for the atomic closure system $\mathcal L_n$ with $n>2$ is given by 

$$\Big(\bigcup\limits_{k=2}^{n-1}\binom{[n]}{k}, \bigcup\limits_{\substack{ i \in [n],  \sigma \subseteq [n] \\ |\sigma|\geq 2}} [\sigma, [n]\setminus \{i\}], \not \in\Big)$$

\end{corollary}

\begin{figure}

\centering

\begin{cxt}
\cxtName{$\not\in$}
\atr{$[ab,U\setminus c]$}
\atr{$[ab,U\setminus d]$}
\atr{$[ac,U\setminus b]$}
\atr{$[ac,U\setminus d]$}
\atr{$[bc,U\setminus a]$}
\atr{$[bc,U\setminus d]$}
\atr{$[abc,U\setminus d]$}
\atr{$[ad,U\setminus b]$}
\atr{$[ad,U\setminus c]$}
\atr{$[bd,U\setminus a]$}
\atr{$[bd,U\setminus c]$}
\atr{$[abd,U\setminus c]$}
\atr{$[cd,U\setminus a]$}
\atr{$[cd,U\setminus b]$}
\atr{$[acd,U\setminus b]$}
\atr{$[bcd,U\setminus a]$}
\obj{..XXXXXXXXXXXXXX}{$ab$}
\obj{XX..XXXXXXXXXXXX}{$ac$}
\obj{XXXX..XXXXXXXXXX}{$bc$}
\obj{X.X.X..XXXXXXXXX}{$abc$}
\obj{XXXXXXX..XXXXXXX}{$ad$}
\obj{XXXXXXXXX..XXXXX}{$bd$}
\obj{.XXXXXXX.X..XXXX}{$abd$}
\obj{XXXXXXXXXXXX..XX}{$cd$}
\obj{XX.XXXX.XXXXX..X}{$acd$}
\obj{XXXX.XXXX.XX.XX.}{$bcd$}
 \end{cxt}   

\caption{The standard context ($10 \times 16$) of the lattice of atomic lattices $\mathcal L_4 $.}\label{fig:scxt4}

\end{figure}

In Fig.~\ref{fig:scxt4}, one can see the standard context ($10 \times 16$) of the lattice of atomic lattices $\mathcal L_4 $ obtained for the base set $U=\{a,b,c,d\}$.

We computed the resulting values of \seqnum{A334254} and \seqnum{A334255} for $n=3,4,5,6$ with our implementation of parallel \textsc{NextClosure} algorithm (originally proposed by Ganter and Reuter~\cite{Ganter:1991}) and thus confirmed the results  of \textsc{\textsc{Atomic AddByOne}}. The total computational time is hard to summarise properly per process, but it took about four days for our approach and five days and 17 hours on a laptop with 12 core Intel i-9 processor in parallel mode (seven days and six hours in sequential mode) for \textsc{NextClosure}.

Another interesting approach for enumeration of atomic lattices on $n$ atoms was implemented by S.\ Mapes in Haskell~\cite{Mapes:2010}; even though its running time for $n=6$ is not reported, it took less than a second for $n=5$. 

If we consider the algorithms and results on the number of Moore families without additional constraints~\seqnum{A193674}, then three notable algorithmic approaches pop up. The first one by Habib and Nourine~\cite{Moorefamilies6} for $n=6$ is based on a bijection between Moore families and ideal color sets of the colored poset based on the sum of $n$ Boolean lattices with $n - 1$ atoms. The second one by Colomb, Irlande, and Raynaud~\cite{Moorefamilies7} for $n=7$ converts the original problem to union-closed sets due to computational efficiency issues, while the third one by Brinkmann and Deklerck~\cite{Brinkmann:2018} uses sophisticated automorphism enumeration techniques to obtain the solution for inequivalent union-closed families in the case $n=7$. The result for $n=5$ was obtained by Higuchi~\cite{Moorefamilies5} using depth-first search. 

An interesting approach by Belohlavek and Vychodil~\cite{Belohlavek:2010} to generate non-isomorphic lattices employs the rejection of candidates to full isomorphism test using the vectors of essential pairs from the corresponding order relation of a lattice. However, this approach is devised for enumeration of lattices with a fixed number of elements.

The direct computation by the \textsc{NextClosure} algorithm is possible up to $n=6$ using the standard context of the lattice of all Moore families as reported by Ganter and Obiedkov~\cite{GanterObiedkov:2017}.

Let us have a look at the ratio of members~\seqnum{A193674} and that of~\seqnum{A334254}.

\begin{table}\label{tbl:ratio}
\centering 
\begin{tabular}{|c|c|c|c|c|c|c|c|}
\hline
$n$ & 0 & 1 & 2 & 3 & 4 & 5 & 6\\
\hline
\seqnum{A193674}, $a(n)$& 1 & 2 & 7 & 61 & 2 480 & 1 385 552 & 75 973 751 474\\

\seqnum{A334254}, $b(n)$ & 1 & 2 & 1 & 8 & 545 & 702 525 & 66 096 965 307\\

\hline
${a(n)}/{b(n)}$ & 1 & 1 & $\approx$ 0.14 & $\approx$0.13 & $\approx$0.22 & $\approx$0.51 & $\approx$0.87\\
\hline

\end{tabular}\caption{Ratio between $n$th members of sequences \seqnum{A193674} and \seqnum{A334254}.}
\end{table}

From Table~\ref{tbl:ratio}, one can see that after reaching minimum 8/61 at $n=3$, the ratio $a(n)/b(n)$ rises up to $\approx$0.87 at $n=6$, which may be a good starting point for analysing asymptotic behaviour of the latter series.

At the moment we propose the following theorem on a weak upper bound for the number of atomic closure systems.

\begin{theorem}\label{thm:upperbound} Let $\mathcal L_n$ and $\mathcal M_n$ be the lattices of all atomic Moore families and all Moore families on a set $[n]$ ($n>1$),  respectively, then

$$|\mathcal L_n| \leq |\mathcal M_n| - 2^n - n \mbox{ .}$$ 

\end{theorem}

\begin{proof} Let us count the size of principal order ideal of $\mathbf 0_{\mathcal L_n}$ taken in $\mathcal M_n$, i.e., at first,  to count all Moore subfamilies of $\{\emptyset, \{1\}, \ldots, \{n\}, [n]\}$. There are $2^n-2$ subsets of $\{\{1\}, \ldots, \{n\}\}$ to be added to $\{\emptyset, [n]\}$ to form a valid Moore family. Every single element set $\{i\} \subset [n]$ gives rise to one out $n$ Moore families in the form $\{\{i\}, [n]\}$. Two remaining Moore families are $\mathbf 0_{\mathcal M_n}=\{[n]\}$ and $\{\emptyset, [n]\}$ per se.
\end{proof}

We conclude with visiting another interesting venue, namely, extremal lattice theory, where questions ``Why finite lattices described by standard contexts are large?'' are studied based on the notion of VC-dimension~\cite{AlbanoChornomaz:2017}. As it was shown by Albano and Chornomaz~\cite{AlbanoChornomaz:2017}, the reason to have a huge number of elements of a lattice is the presence in its standard contexts of the so-called contranominal scales, i.e.\ induced subcontexts (subrelations) of the form $\mathbb N^c(k):=(\{1,\ldots,k\},\{1,\ldots,k\},\neq)$ (e.g., the rightmost binary relation in Fig.~\ref{fig:scales} is $\mathbb N^c(3)$).

For example, the closure systems generated by a contranominal scale on $n$  elements taken as a standard context has $2^n$ closed sets.

 The \textit{breadth} of a complete lattice is the number of atoms of the largest Boolean lattice that the lattice contains as a suborder, i.e.\ the size of the largest contranominal scale subrelation of its standard context (valid for all finite contexts)~as noted by Ganter~\cite{Ganter:2020}.

\begin{theorem}[Albano and Chornomaz~\cite{AlbanoChornomaz:2017}]\label{thm:ACbound}
Let $\mathbb K:= ([n],U,I\subseteq [n]\times U)$ be an $\mathbb N^c(k)$-free formal context, then
$$|\mathbf L(\mathbb K)| \leq \sum\limits_{i=0}^{k-1} \binom{n}{i}.$$ 
\end{theorem}

In what follows, we denote the sum from Theorem~\ref{thm:ACbound} by $f_{AC}(n,k)$.

By using this theorem and our knowledge on the number of meet-irreducible element for the lattice of all atomic Moore families and Moore families, respectively, we obtain two series (Table~\ref{tbl:breadth}) for the estimated breadth of those two lattices for known $|\mathcal L_n|$ and $|\mathcal M_n|$. Note that $|J(\mathcal M_n)|=2^n-1$ follows from Definition 17 and Proposition 18 in Caspard and Monjardet~\cite{Caspard:2003}.

\begin{table}
\centering 
\begin{tabular}{|c|c|c|c|c|c|}
\hline
$n$ &  3 & 4 & 5 & 6 & 7\\
\hline
$|J(\mathcal L_n)|$ & 3 & 10 & 25 & 56 & 119\\
$|J(\mathcal M_n)| $ & 7 & 15 & 31 & 63 & 127\\
Estimated breadth of $\mathcal L_n$& 3 & 5 & 7 & 11 & ?\\
Estimated breadth of $\mathcal M_n$  & 3 & 5 & 7 & 10 & 16\\
\hline
Breadth of $\mathcal L_n$ & 3 & 7 & 13 &  $\geq$ 18 & $\geq$ 25\\
\hline

\end{tabular}\caption{Estimated breadths of lattices $\mathcal L_n$ and $\mathcal M_n$ for $n=3$ up to 6 and 7, respectively.}\label{tbl:breadth}
\end{table}

For example, for $n=6$ we know $|J(\mathcal L_6)|=56$. For $k=11$ we get $f_{AC}(56,11)=\sum\limits_{i=0}^{10} \binom{56}{i}=44 872 116 214 < | \mathcal L_6|=66 096 965 307$, while for $k=12$ we have  $f_{AC}(n,k)(56,12)=193 774 331 494 > |\mathcal L_6|$. So, $\mathcal L_6$ contains a $2048$-elements Boolean lattice and  the breadth of $\mathcal L_6$ is at least 11.

The actual breadth of $\mathcal L_3$ is indeed 3 since the standard context of  $\mathcal L_3$ coincides with $\mathbb N^c(3)$, while the actual breadth of $\mathcal L_4$ is 7  since there are 80 embedded $\mathbb N^c(7)$ and that of  $\mathcal L_5$ is 13  since there are 10980 embedded $\mathbb N^c(13)$ (the last line in Table~\ref{tbl:breadth} is found by a combinatorial search).

\section{Acknowledgments}

We would like to thank Bernhard Ganter, Sergei Kuznetsov, and Sergei Obiedkov for introducing us to the main algorithms for closed sets enumeration.  We also thank Florent Domenach for discussions on lattice cryptomorphisms, Radim Belohlavek, Leonard Kwuida, Bogdan Chornomaz, Lhouari Nourine, Dmitry Piontkowski, and Marcin Wolski for several fruitful discussions on related topics, and last but not least we thank Alexei Neznanov for the first familiarity with OEIS and graph automorphisms enumeration.

The paper was prepared within the framework of the HSE University Basic Research Program and RFBR (Russian Foundation for Basic Research) according to the research project No 19-29-01151. This research was supported in part through computational resources of HPC facilities at HSE University. We would also to thank Gennady Khvorykh for managing our access to the computational server at the Institute of Molecular Genetics of National Research Centre ``Kurchatov Institute''.

\bigskip
\hrule
\bigskip

\noindent 2010 {\it Mathematics Subject Classification}: Primary
06-04; Secondary 06A07, 06A15.

\noindent \emph{Keywords:} enumeration, Moore set, atomic lattice.

\bigskip
\hrule
\bigskip

\noindent (Concerned with sequences
\seqnum{A334254},
\seqnum{A334255},
\seqnum{A235604},
\seqnum{A355517},
\seqnum{A193674} and
\seqnum{A305233}.)

\end{document}